\documentclass{article}
\usepackage{graphicx,color}
\usepackage{amsfonts}
\usepackage{amssymb}
\usepackage{amsmath}
\usepackage{epstopdf}

\newtheorem{theorem}{Theorem}

\newtheorem{definition}[theorem]{Definition}

\newtheorem{lemma}[theorem]{Lemma}

\newenvironment{proof}[1][Proof]{\textbf{#1.} }{\ \rule{0.5em}{0.5em}}

\definecolor{green}{rgb}{0.00,0.50,0.00}

\begin{document}

\title{Quantum Filtering in Coherent States}
\author{John E. Gough$^a$, Claus K\"{o}stler$^b$ \\
Institute of Mathematics and Physics, Aberystwyth University, SY23 3BZ, Wales, United Kingdom \\
a) \textbf{jug@aber.ac.uk}, \qquad b) \textbf{cck@aber.ac.uk} }
\date{}



\maketitle
\begin{abstract}
We derive the form of the Belavkin-Kushner-Stratonovich equation describing the filtering of a continuous observed quantum system via non-demolition measurements when the statistics of the input field used for the indirect measurement are in a general coherent state.
\end{abstract}

\begin{center}
\textbf{Dedicated to Robin Hudson on the occasion his 70th birthday.}
\end{center}

\section{Introduction}
One of the most remarkable consequences of Hudson-Parthasarathy quantum stochastic calculus \cite{HP} is Belavkin's formulation of a quantum theory of filtering based on non-demolition measurements of an output field that has interacted with a given system \cite{Bel80,Bel92,Bel92a,Bel97}. Specifically, we must measure a particular feature of the field, for instance a field quadrature or the count of the field quanta, and this determines a self-commuting, therefore \emph{essentially classical}, stochastic process. The resulting equations have structural similarities with the classical analogues appearing in the work of Kallianpur, Striebel, Kusnher, Stratonovich, Zakai, Duncan and Mortensen on nonlinear filtering, see \cite{DavisMarcus,Kush79,Kush80,Zak69}. AS a consequence, earlier models of repeated quantum photon counting measurements developed by Davies \cite{Dav69,Dav76} could be realized through a concrete theory: this was first shown by taking the pure-jump process limit of diffusive quantum filtering problems \cite{BarBel91}. 

\bigskip
\indent
There has been recent interest amongst the physics community in quantum filtering as an applied technique in quantum feedback and control \cite{AASDM02,Bar90,BouEdBel05,BouGutMaa04,BvH Ref Prob,DHJMT00,GouBelSmol,Jam04,vHSM05,WM93}. An additional driver is the desire to go beyond the situation of a vacuum field and derive the filter for other physically important states such as  thermal, squeezed, single photon states, etc. In this note we wish to present the filter for non-demolition quadrature and photon-counting measurements when the choice of state for the input field is a coherent state with intensity function $\beta$. The resulting filters are a deformation of the vacuum filters and reduce to the latter when we take $\beta \equiv 0$, this is perhaps to be expected given that the coherent states have a continuous-in-time tensor product factorization property. We derive the filters using the reference probability approach, as well as the characteristic function approach.

\subsection{Classical Non-linear Filtering}

We consider a state based model where the state $X_{t}$ evolves according to
a stochastic dynamics and we make noisy observations $Y_{t}$ on the state.
The dynamics-observations equations are the SDEs 
\begin{eqnarray}
dX_{t} &=&v\left( X_{t}\right) dt+\sigma _{X}\left( X_{t}\right) dW_{t}^{%
\mathrm{proc}},  \label{pairsys} \\
dY_{t} &=&h\left( X_{t}\right) dt+\sigma _{Y}dW_{t}^{\mathrm{obs}}
\label{pairobs}
\end{eqnarray}
and we assume that the process noise $W^{\mathrm{proc}}$ and the observation
noise $W_{\mathrm{obs}}$ are uncorrelated multi-dimensional Wiener
processes. The generator of the state diffusion is then 
\begin{equation*}
\mathcal{L}=v^{i}\partial _{i}+\frac{1}{2}\Sigma _{XX}^{ij}\partial _{ij}^{2}
\end{equation*}
where $\Sigma _{XX}=\sigma _{X}\sigma _{X}^{\top }$. The aim of filtering
theory to obtain a least squares estimate for the state dynamics. More
specifically, for any suitable function $f$ of the state, we would like to
evaluate the conditional expectation 
\begin{equation*}
\pi _{t}\left( f\right) :=\mathbb{E}[ f\left( X_{t}\right) | \mathcal{F}_{t]}^{Y}],
\end{equation*}
with $\mathcal{F}_{t]}^{Y}$ being the $\sigma $-algebra generated by the
observations up to time $t$.

\subsubsection{Kallianpur-Striebel Formula}

By introducing the \emph{Kallianpur-Striebel likelihood function }
\begin{equation*}
L_{t}\left( \mathbf{x}|\mathbf{y}\right) =\exp \int_{0}^{t}\left\{ h\left(
x_{s}\right) ^{\top }dy_{s}-\frac{1}{2}h\left( x_{s}\right) ^{\top }h\left(
x_{s}\right) ds\right\}
\end{equation*}
for sample state path $\mathbf{x}=\left\{ x_{s}:0\leq s\leq t\right\} $
conditional on a given sample observation $\mathbf{y}=\left\{ y_{s}:0\leq s\leq t\right\} $%
, we may represent the conditional expectation as 
\begin{equation*}
\pi _{t}\left( f\right) =\left. \frac{\int_{C_{0}^{x}\left[ 0,t\right]
}f\left( x_{t}\right) L_{t}\left( \mathbf{x}|\mathbf{y}\right) \mathbb{P}[d%
\mathbf{x}]}{\int_{C_{0}^{x}\left[ 0,t\right] }L_{t}\left( \mathbf{x}|%
\mathbf{y}\right) \mathbb{P}[d\mathbf{x}]}\right| _{\mathbf{y}=Y\left(
\omega \right) }=\frac{\sigma _{t}\left( f\right) }{\sigma _{t}\left(
1\right) }
\end{equation*}
where $\mathbb{P}$ is canonical Wiener measure and 

\begin{equation*}
\sigma _{t}\left( f\right) \left( \omega \right) =\int_{C_{0}^{x}\left[ 0,t%
\right] }f\left( x_{t}\right) L_{t}\left( \mathbf{x}|Y\left( \omega \right)
\right) \mathbb{P}\left[ d\mathbf{x}\right] .
\end{equation*}

\subsubsection{Duncan-Mortensen-Zakai and Kushner-Stratonovich Equations}

Using the It\={o} calculus, we may obtain the\emph{ Duncan-Mortensen-Zakai
equation} for the un-normalized filter $\sigma _{t}\left( f\right) $, and the
\emph{Kushner-Stratonovich} equation for the normalized version $\pi _{t}\left(
f\right) $. These are 
\begin{eqnarray*}
d\sigma _{t}\left( f\right) &=&\sigma _{t}\left( \mathcal{L}f\right)
dt+\sigma _{t}\left( fh^{\top }\right) dY_{t}, \\
d\pi _{t}\left( f\right) &=&\pi _{t}\left( \mathcal{L}f\right) dt+\left[ \pi
_{t}\left( fh^{\top }\right) -\pi _{t}\left( f\right) \pi _{t}\left( h^{\top
}\right) \right] dI_{t},
\end{eqnarray*}
where $\left( I_{t}\right) $ are the innovations: 
\begin{equation*}
dI_{t}:=dY_{t}-\pi _{t}\left( h\right) dt,\;I\left( 0\right) =0.
\end{equation*}
We note that there exist variants of these equations for more general
processes than diffusions (in particular for point processes which will be of relevance for photon counting), and for the case where the process and
observation noises are correlated.

\subsubsection{Pure versus Hybrid Filtering Problems}

We remark that we follow the traditional approach of adding direct Wiener
noise $W_{\mathrm{obs}}$ to the observations. We could of course consider a
more general relation of the form $dY_{t}=h\left( X_{t}\right) dt+\sigma
_{Y}dW_{t}^{\mathrm{obs}}$ but for constant coefficients $\sigma _{Y}$ a
simple rearrangement returns us to the above setup.

\bigskip

The situation where we envisage $dY_{t}=h\left( X_{t}\right) dt+\sigma
_{Y}\left( X_{t}\right) dW_{t}^{\mathrm{obs}}$, with $\sigma _{Y}$ a known
function of the unobserved state, must be considered as being \textit{too
good to be true} since we can then obtain information about the unobserved
state by just examining the quadratic variation of the observations process,
since we then have $dY\left( t\right) dY\left( t\right) ^{\top }=\sigma
_{Y}(X_{t})\sigma _{Y}\left( X_{t}\right) ^{\top }dt$. For instance, in the
case of scalar processes, if we have $\sigma _{Y}\left( X\right) =\gamma |X|$
then knowledge of the quadratic variation yields the magnitude $|X_{t}|$ of
the signal without any need for filtering. Such situations rarely if
ever arise in practice, and one naturally restricts to pure filtering problems.

\section{Quantum Filtering}

We wish to describe the quantum mechanical analogue of the classical
filtering problem. To begin with, we note that in quantum theory the physical
degrees of freedom are modeled as observables, that is self-adjoint
operators on a fixed Hilbert space $\mathfrak{h}$. The observables will
generally not commute with each other. In place of the classical notion of a
state, we will have a normalized vector $\psi \in \mathfrak{h}$ and the averaged
of an observable $X$ will be give by the real number $\langle \psi |X\psi
\rangle $. (Here we following the physicist convention of taking the inner
product $\langle \psi |\phi \rangle $ to be linear in the second argument $%
\phi $ and conjugate linear in the first $\psi $.) More generally we define
a quantum state to be a positive, normalized linear functional $\mathbb{E}$
on the set of operators. Every such expectation may be written as 
\begin{equation*}
\mathbb{E}\left[ X\right] =\mathrm{tr}_{\mathfrak{h}}\left\{ \varrho X\right\}
\end{equation*}
where $\varrho $ is a positive trace-class operator normalized so that tr$_{%
\mathfrak{h}}\varrho =1$. The operator $\varrho $ is referred to as a \textit{%
density matrix}. The set of all states is a convex set whose extreme points
correspond to the density matrices that are rank-one projectors onto the
subspace spanned by a unit vectors $\psi \in \mathfrak{h}$.

\bigskip

To make a full analogy with classical theory, we should exploit the
mathematical framework of quantum probability which gives the appropriate
generalization of probability theory and stochastic processes to the quantum
setting. The standard setting is in terms of a von Neumann algebra of
observables over a fixed Hilbert space, which will generalize the notion of
an algebra of bounded random variables, and take the state to be an
expectation functional which is continuous in the normal topology. The
latter condition is equivalent to the $\sigma $-finiteness assumption in
probability theory and results in all the states of interest being
equivalent to a density matrix.

\bigskip

In any given experiment, we may only measure commuting observables. Quantum
estimation theory requires that the only observables that we may estimate
based on a particular experiment are those which commute with the measured
observables. In practice, we do not measure a quantum system directly, but
apply an input field and measure a component of the output field. The input
field results in a open dynamics for the system while measurement of the
output ensures that we met so called non-demolition conditions which
guarantee that quantum measurement process itself does not destroy the
statistical features which we would like to infer. We will now describe
these elements in more detail below.

\subsection{Quantum Estimation}

We shall now describe the reference probability approach to quantum
filtering. Most of our conventions following the presentation of Bouten and
van Handel \cite{BvH Ref Prob}.

\bigskip

Let $\mathfrak{A}$ be a von Neumann algebra and $\mathbb{E}$ be a normal state. In a given experiment one may only measure a set of \textit{commuting} observables $\{Y_{\alpha }:\alpha \in A\}$. Define the measurement algebra to be the commutative von Neumann algebra generated by the chosen observables 
\begin{equation*}
\mathfrak{M}=\mathrm{vN}\{Y_{\alpha }:\alpha \in A\}\subset \mathfrak{A}.
\end{equation*}

We may estimate an observable $X\in \mathfrak{A}$ from an experiment with measurement algebra $\mathfrak{M}$ if and only if 
\begin{equation*}
X\in \mathfrak{M}^{\prime }:= \{A\in \mathfrak{A}:[A,Y]=0,\,\forall \,Y\in \mathfrak{M}%
\},
\end{equation*}
That is, if it is physically possible to measure $X$ in addition to all the $Y_{\alpha }$. Therefore the algebra vN$\{X,Y_{\alpha }:\alpha \in A\}$ must again be commutative. We may then set about defining the conditional expectation of estimable observables onto the measurement algebra.

\begin{definition}
For commutative von Neumann algebra $\mathfrak{M}$, the conditional expectation onto $\mathfrak{M}$\ is the map 
\[
\mathbb{E}[\cdot \mid \mathfrak{M}]:\mathfrak{M}^{\prime }\mapsto \mathfrak{M}
\]
by 
\begin{equation}
\mathbb{E}[\,\mathbb{E}[X \mid \mathfrak{M}]\,Y]\equiv \mathbb{E}[XY],\,\forall
Y\in \mathfrak{M}.
\end{equation}
\end{definition}

\bigskip

In contrast to the general situation regarding conditional expectations in the non-commutative setting of von Neumann algebras \cite{Tak71}, this particular definition is always nontrivial insofar as existence is guaranteed. Introducing the norm $\Vert A\Vert ^{2}:=\mathbb{E}[A^{\dag }A]$, we see that the conditional expectation always exists and is unique up to norm-zero terms. It moreover satisfies the least squares property 
\begin{equation*}
\Vert X-\mathbb{E}[X\mid \mathfrak{M}]\Vert \leq \Vert X-Y\Vert ,\quad \forall
Y\in \mathfrak{M}.
\end{equation*}
As the set vN$\left\{ X,Y_{\alpha }:\alpha \in A\right\} $ is a commutative von Neumann algebra for each $X\in \mathfrak{M}^{\prime }$, it will be isomorphic to the space of bounded functions on a measurable space by Gelfand's theorem. The state induces a probability measure on this space and we may obtain the standard conditional expectation of the random variable corresponding to $X$ onto the $\sigma $-algebra generated by the functions corresponding to the $Y_{\alpha }$. This classical conditional expectation then corresponds to a unique element $\mathbb{E}[X\mid \mathfrak{M}]\in \mathfrak{M}$ and this gives the construction of the quantum conditional expectation. 

\bigskip

We sketch the conditional expectation in figure 1.  Note that while this may seem trivial at first sight, it should be stressed that the commutant $\mathfrak{M}^{\prime }$ itself will typically be a non-commutative algebra, so that while our measured observables commute, and what we wish to estimate must commute with our measured observables, the object we can estimate need not commute amongst themselves.

\begin{figure}[tbph]
\centering
\setlength{\unitlength}{.04cm} 
\begin{picture}(120,60)
\label{pic2}
\thicklines
\put(60,25){\oval(100,60)}
\put(60,25){\oval(80,20)}
\thinlines
\put(70,45){\vector(0,-1){22}}
\put(70,45){\circle*{2}}
\put(70,21.5){\circle*{2}}
\put(60,42){$X$}
\put(38,20){$\mathbb{E} [X | \mathfrak{M} ]$}
\put(107,50){$\mathfrak{M}'$}
\put(97,35){$\mathfrak{M}$}
\end{picture}
\caption{Quantum Conditional Expectation $\mathbb{E} [ \cdot | \mathfrak{M} ]$}
\end{figure}
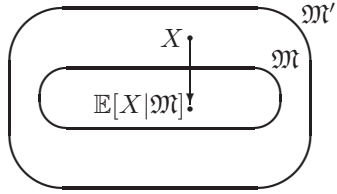

The following two lemmas will be used extensively, see \cite{BvH Ref Prob} and \cite{BvHJ}.

\begin{lemma}[Unitary rotations]
\label{unitary rotations} Let $U$ be unitary and define $\mathbb{\tilde{E}}%
\left[ X\right] := \mathbb{E}\left[ U^{\dag }XU\right] $ and let $%
\mathfrak{\tilde{M}}=U^{\dag }\mathfrak{M}U$. Then 
\[
\mathbb{E} [ U^{\dag }XU \mid \mathfrak{\tilde{M}} ] =U^{\dag }\mathbb{\tilde{E}}%
\left[ X\mid \mathfrak{M}\right] U.
\]
\end{lemma}

Here we think of going from the Schr\"{o}dinger picture where the state $\mathbb{E}$ is fixed and observables evolve to $U^{\ast }XU$, to the Heisenberg picture where the state evolves to $\mathbb{\tilde{E}}$ and the  observables are fixed. Lemma \ref{unitary rotations} tells us how we may transform the conditional expectation between these two pictures.

\begin{lemma}[Quantum Bayes' formula]
\label{QBayes} Let $F\in \mathfrak{M}^{\prime }$ with $\mathbb{E}\left[ F^{\dag
}F\right] =1$ and set $\mathbb{E}_{F}\left[ X\right] := \mathbb{E}
\left[ F^{\dag }XF\right] $. Then 
\[
\mathbb{E}_{F}\left[ X\mid \mathfrak{M}\right] =\frac{\mathbb{E}\left[ F^{\dag
}XF \mid\mathfrak{M}\right] }{\mathbb{E}\left[ F^{\dag }F\mid \mathfrak{M}\right] }.
\]
\end{lemma}

\begin{proof}
For all $Y\in \mathfrak{M}$, 
\begin{eqnarray*}
\mathbb{E}\left[ \mathbb{E}\left[ F^{\dag }XF\mid \mathfrak{M}\right] Y\right]
&=&\mathbb{E}\left[ F^{\dag }XFY\right] \\
&=&\mathbb{E}_{F}[XY],\qquad \mathrm{since}\left[ F,Y\right] =0, \\
&=&\mathbb{E}_{F}\left[ \mathbb{E}_{F}\left[ X\mid \mathfrak{M}\right] Y\right]
\\
&=&\mathbb{E}\left[ F^{\dag }F\mathbb{E}_{F}\left[ X\mid \mathfrak{M}\right] Y%
\right] ,\;\mathrm{since\;}F\in \mathfrak{M}^{\prime } \\
&=&\mathbb{E}\left[ \mathbb{E}\left[ F^{\dag }F\mid \mathfrak{M}\right] \mathbb{E%
}_{F}\left[ X\mid \mathfrak{M}\right] Y\right] .
\end{eqnarray*}
\end{proof}

Note that the proof only works if $F\in \mathfrak{M}^{\prime }$!

\subsection{Quantum Stochastic Processes}

We begin by reviewing the theory of quantum stochastic calculus developed by Hudson and Parthasarathy \cite{HP} which gives the mathematical framework with which to generalize the notions of the classical It\={o} integration theory.

\bigskip

We take $\mathbb{R}_{+}=[0,\infty )$. We shall denote by $L_{\mathrm{symm} }^{2}(\mathbb{R}_{+}^{n})$ the space of all square-integrable functions of $n $ positive variables that are completely symmetric: that is, invariant under interchange of any pair of its arguments. The Bose Fock space over $L^{2}(\mathbb{R}_{+})$ is then the infinite direct sum Hilbert space 
\begin{equation*}
\mathfrak{F}:= \bigoplus_{n=0}^{\infty }L_{\mathrm{symm}}^{2}(\mathbb{R}_{+}^{n})
\end{equation*}
with the $n=0$ space identified with $\mathbb{C}$. An element of $\mathfrak{F}$
is then a sequence $\Psi =\left( \psi _{n}\right) _{n=0}^{\infty }$ with $%
\psi _{n}\in L_{\mathrm{symm}}^{2}(\mathbb{R}_{+}^{n})$ and $\left\| \Psi
\right\| ^{2}=\sum_{n=0}^{\infty }\int_{[0,\infty )^{n}}|\psi _{n}\left(
t_{1},\cdots ,t_{n}\right) |^{2}$ $dt_{1}\cdots dt_{n}<\infty $. Moreover,
the Fock space has inner product 
\begin{equation*}
\langle \Psi \mid \Phi \rangle =\sum_{n=0}^{\infty }\int_{[0,\infty
)^{n}}\psi _{n}\left( t_{1},\cdots ,t_{n}\right) ^{\ast }\phi \left(
t_{1},\cdots ,t_{n}\right) .
\end{equation*}
Th physical interpretation is that $\Psi =\left( \psi _{n}\right)
_{n=0}^{\infty }\in \mathfrak{F}$ describes the state of a quantum field consisting of an indefinite number of indistinguishable (Boson) particles on the half-line $\mathbb{R}_{+}$. A simple example is the \textit{vacuum vector} defined by  
\begin{equation*}
\Omega := \left( 1,0,0,\cdots \right)
\end{equation*}
clearly corresponding to no particles. (Note that the no-particle state is a genuine physical state of the field and is not just the zero vector of $\mathfrak{F}$!) An important class of vectors are the \textit{coherent states} $\Psi \left( \beta \right) $ defined by 
\begin{equation*}
\left[ \Psi \left( \beta \right) \right] _{n}\left( t_{1},\cdots
,t_{n}\right) := e^{-\left\| \beta \right\| ^{2}}\frac{1}{\sqrt{n!}}\beta
\left( t_{1}\right) \cdots \beta (t_{n}),
\end{equation*}
for $\beta \in L^{2}(\mathbb{R}_{+})$. (The $n=0$ component understood as $ e^{-\left\| \beta \right\| ^{2}}$.) The vacuum then corresponds to $\Psi
\left( 0\right) $.

\bigskip

For each $t>0$ we define the operators of annihilation $B\left( t\right) $, creation $B^{\ast }\left( t\right) $ and gauge $\Lambda \left( t\right) $ by 
\begin{eqnarray*}
\left[ B\left( t\right) \Psi \right] _{n}\left( t_{1},\cdots ,t_{n}\right)
&:= &\sqrt{n+1}\int_{0}^{t}\psi _{n+1}\left( s,t_{1},\cdots ,t_{n}\right) ds,
\\
\left[ B^{\ast }\left( t\right) \Psi \right] _{n}\left( t_{1},\cdots
,t_{n}\right) &:= &\frac{1}{\sqrt{n}}\sum_{j=1}^{n}1_{[0,t]}(t_{j})\psi
_{n-1}\left( t_{1},\cdots ,\widehat{t_{j}},\cdots ,t_{n}\right) , \\
\left[ \Lambda \left( t\right) \Psi \right] _{n}\left( t_{1},\cdots
,t_{n}\right) &:= &\sum_{j=1}^{n}1_{[0,t]}(t_{j})\psi _{n}\left(
t_{1},\cdots ,t_{n}\right) .
\end{eqnarray*}
The creation and annihilation process are adjoint to each other and the gauge is self-adjoint. We may define a field quadrature by 
\begin{equation*}
Q\left( t\right) =B\left( t\right) +B^{\ast }\left( t\right)
\end{equation*}
and this yields a quantum stochastic process which is essentially classical in the sense that it is self-adjoint and self-commuting, that is $\left[ Q\left( t\right) ,Q\left( s\right) \right] =0$ for all $t,s\geq 0$. We remark that for each $\theta \in \lbrack 0,2\pi )$ we may define quadratures $Q_{\theta }\left( t\right) =e^{-i\theta }B\left( t\right) +e^{i\theta
}B^{\ast }\left( t\right) $ which again yield essentially classical processes, however different quadratures will not commute! For the choice of the vacuum state, $\left\{ Q\left( t\right) :t\geq 0\right\} $ then yields a representation of the Wiener process: for real $k\left( \cdot \right) $
\begin{equation*}
\langle \Omega \mid e^{i\int_{0}^{\infty }k\left( t\right) dQ\left( t\right)
}\Omega \rangle =e^{-\frac{1}{2}\int_{0}^{\infty }k\left( t\right) ^{2}dt}.
\end{equation*}

We also note that $\left\{ \Lambda \left( t\right) :t\geq 0\right\} $ is also an essentially classical process and for the choice of a coherent state
yields a non-homogeneous Poisson process: for real $k\left( \cdot \right) $ 
\begin{equation*}
\langle \Psi \left( \beta \right) \mid e^{i\int_{0}^{\infty }k\left(
t\right) d\Lambda \left( t\right) }\Psi \left( \beta \right) \rangle =\exp
\int_{0}^{\infty }|\beta \left( t\right) |^{2}\left( e^{ik\left( t\right)
}-1\right) dt.
\end{equation*}

We consider a quantum mechanical system with Hilbert space $\mathfrak{h}$ being driven by an external quantum field input. the quantum field will be modeled as an idealized Bose field with Hilbert space $\Gamma \left( L^{2}\left(  \mathbb{R}_{+},dt\right) \right) $ which is the Fock space over the one-particle space $L^{2}\left( \mathbb{R}_{+},dt\right) $. Elements of the Fock space may be thought of as vectors $\Psi =\oplus _{n=0}^{\infty }\psi _{n}$ where $\psi _{n}=\psi _{n}\left( t_{1},\cdots ,t_{n}\right) $ is a completely symmetric functions with 
\begin{equation*}
\sum_{n=0}^{\infty }\int_{[0,\infty )^{n}}|\psi _{n}\left( t_{1},\cdots
,t_{n}\right) |^{2}dt_{1}\cdots dt_{n}<\infty .
\end{equation*}

The Hudson-Parthasarathy theory of quantum stochastic calculus gives a generalization of the It\={o} theory of integration to construct integral processes with respect to the processes of annihilation, creation, gauge and, of course, time. This leads to the quantum It\={o} table \ref{QIT}.

\begin{table}[tbp]
\caption{Quantum It\={o} Table}
\label{QIT}
\begin{center}
{\footnotesize 
\begin{tabular}{l|llll}
$\times $ & $dB$ & $d\Lambda $ & $dB^{\ast }$ & $dt$ \\ \hline
$dB$ & 0 & $dB$ & $dt$ & 0 \\ 
$d\Lambda $ & 0 & $d\Lambda $ & $dB^{\ast }$ & 0 \\ 
$dB^{\ast }$ & 0 & 0 & 0 & 0 \\ 
$dt$ & 0 & 0 & 0 & 0
\end{tabular}
. }
\end{center}
\end{table}

We remark that the Fock space carries a natural filtration in time obtained from the decomposition $\mathfrak{F}\cong \mathfrak{F}_{t]}\otimes \mathfrak{F}_{(t}$ into past and future subspaces: these are the Fock spaces over $L^{2}\left[ 0,t\right] $ and $L^{2}(t,\infty )$ respectively.

\subsection{Continuous-Time Quantum Stochastic Evolutions}

On the joint space $\mathfrak{h}\otimes \mathfrak{F}$, we consider the quantum stochastic process $V(\cdot )$ satisfying the QSDE 
\begin{equation*}
dV(t)=\left\{ 
\left( S-I\right) \otimes d\Lambda (t)+L\otimes dB^{\ast }\left( t\right) 
-L^{\ast }S\otimes dB\left( t\right) -(\frac{1}{2}L^{\ast }L+iH)\otimes dt
\right\} V(t),
\end{equation*}with $V\left( 0\right) =1$, and where $S$ is unitary, $L$ is bounded and $H$ self-adjoint. This specific form of QSDE may be termed the \emph{Hudson-Parthasarathy equation} as the algebraic conditions on the coefficients are necessary and sufficient to ensure unitarity (though the restriction for $L$ to be bounded can be lifted). The process is also adapted in the sense that for each $t>0$, $V\left( t\right) $ acts non-trivially on the component $\mathfrak{h}\otimes \mathfrak{F}_{t]}$ and trivially on $\mathfrak{F}_{(t}$.

\subsubsection{The Heisenberg-Langevin Equations}

For a fixed system operator $X$ we set 
\begin{equation}
j_{t}\left( X\right) :=V^{\dag }\left( t\right) \left[ X\otimes I\right]
V\left( t\right) .
\end{equation}
Then from the quantum It\={o} calculus we get 
\begin{equation}
dj_{t}\left( X\right)  =j_{t}\left( \mathcal{L}_{11}X\right) \otimes
d\Lambda \left( t)+j_{t}\left( \mathcal{L}_{10}X\right) \otimes dB^{\ast
}\left( t\right) \right)  
+j_{t}\left( \mathcal{L}_{01}X\right) \otimes dB\left( t\right)
+j_{t}\left( \mathcal{L}_{00}X\right) \otimes dt  \label{dynamical}
\end{equation}
where the Evans-Hudson maps $\mathcal{L}_{\mu \nu }$ are explicitly given by
\begin{eqnarray*}
\mathcal{L}_{11}X &=&S^{\ast }XS-X, \\
\mathcal{L}_{10}X &=&S^{\ast }[X,L], \\
\mathcal{L}_{01}X &=&[L^{\ast },X]S \\
\mathcal{L}_{00}X &=&\mathcal{L}_{\left( L,H\right) }
\end{eqnarray*}
and in particular $\mathcal{L}_{00}$ takes the generic form of a Lindblad
generator: 
\begin{equation}
\mathcal{L}_{\left( L,H\right) }=\frac{1}{2}L^{\ast }[X,L]+\frac{1}{2}[%
L^{\ast },X]L-i\left[ X,H\right] .
\end{equation}

\subsubsection{Output Processes}

We introduce the processes 
\begin{eqnarray}
B^{\mathrm{out}}\left( t\right) &:= & V^{\dag }\left( t\right) \left[
I\otimes B\left( t\right) \right] V\left( t\right) ,  \notag \\
\Lambda ^{\mathrm{out}}\left( t\right) &:= &V^{\dag }\left( t\right) \left[
I\otimes \Lambda \left( t\right) \right] V\left( t\right) .
\end{eqnarray}
We note that we equivalently have $B^{\mathrm{out}}\left( t\right) \equiv
V^{\dag }\left( T\right) \left[ 1\otimes B\left( t\right) \right] V\left(
T\right) $, for $t\leq T$. Again using the quantum It\={o} rules, we see
that 
\begin{eqnarray}
dB_{\mathrm{out}} &=&j_{t}(S)dB(t)+j_{t}(L)dt,  \notag \\
d\Lambda _{\mathrm{out}} &=&d\Lambda \left( t\right) +j_{t}(L^{\ast
}S)dB^{\ast }(t)+j_{t}(S^{\ast }L)dB(t)+j_{t}(L^{\ast }L)dt.  \label{output}
\end{eqnarray}

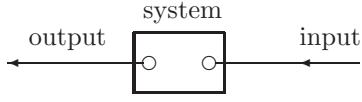
\begin{figure}[tbp]
\centering
\setlength{\unitlength}{.04cm} 
\begin{picture}(120,45)
\label{pic1}
\thicklines
\put(45,10){\line(0,1){20}}
\put(45,10){\line(1,0){30}}
\put(75,10){\line(0,1){20}}
\put(45,30){\line(1,0){30}}
\thinlines
\put(48,20){\vector(-1,0){45}}
\put(120,20){\vector(-1,0){20}}
\put(120,20){\line(-1,0){48}}
\put(50,20){\circle{4}}
\put(70,20){\circle{4}}
\put(100,26){input}
\put(48,35){system}
\put(10,26){output}
\end{picture}
\caption{Input - Output component}
\end{figure}

\subsubsection{The Measurement Algebra}

We wish to consider the problem of continuously measuring a quantum
stochastic process associated with the output field. we shall chose to
measure an observable process of the form 
\begin{equation}
Y^{\mathrm{out}}(t):= V(t)^{\dag }[I\otimes Y^{\mathrm{in}}(t)]V(t)
\end{equation}
which corresponds to a quadrature of the field when 
\begin{equation*}
Y^{\mathrm{in}}(t)=Q\left( t\right) =B^{\ast }(t)+B(t),
\end{equation*}
or counting the number of\ output photons when 
\begin{equation*}
Y^{\mathrm{in}}(t)=\Lambda (t).
\end{equation*}

We introduce von Neumann algebra 
\begin{equation*}
\mathfrak{Y}_{t]}^{\mathrm{in}}=\mathrm{vN}\left\{ Y^{\mathrm{in}}\left(
s\right) :0\leq s\leq t\right\} ,
\end{equation*}
and define the measurement algebra up to time $t$ to be 
\begin{equation}
\mathfrak{Y}_{t]}^{\mathrm{out}}=\mathrm{vN}\left\{ Y^{\mathrm{out}}\left(
s\right) :0\leq s\leq t\right\} \equiv V\left( t\right) ^{\dag }\mathfrak{Y}%
_{t]}^{\mathrm{in}}V\left( t\right) .
\end{equation}

Note that both algebras are commutative:

\begin{equation*}
\left[ Y^{\mathrm{out}}\left( t\right) ,Y^{\mathrm{out}}\left( s\right) %
\right] =V\left( T\right) ^{\dag }\left( I\otimes \left[ Y^{\mathrm{in}%
}\left( t\right) ,Y^{\mathrm{in}}\left( s\right) \right] \right) V\left(
T\right) =0
\end{equation*}
for $T=t\vee s$. The family $\left\{ \mathfrak{Y}_{t]}^{\mathrm{out}}:t\geq
0\right\} $ then forms an increasing family (\textit{filtration)} of von
Neumann algebras.

\subsubsection{The Non-Demolition Property}

The system observables may be estimated from the current measurement algebra 
\begin{equation}
j_{t}\left( X\right) \in \left( \mathfrak{Y}_{t]}^{\mathrm{out}}\right) ^{\prime
}.  \label{non-dem}
\end{equation}
The proof follows from the observation that for $t\geq s$ 
\begin{equation*}
\left[ j_{t}\left( X\right) ,Y^{\mathrm{out}}\left( s\right) \right]
=V\left( t\right) ^{\dag }\left[ X\otimes I,I\otimes Y^{\mathrm{in}}\left(
s\right) \right] V\left( t\right) =0.
\end{equation*}

\subsection{Constructing The Quantum Filter}

The filtered estimate for $j_{t}\left( X\right) $ given the measurements of
the output field is then 
\begin{equation*}
\pi _{t}\left( X\right) := \mathbb{E}\left[ j_{t}\left( X\right) \mid \mathfrak{Y%
}_{t]}^{\mathrm{out}}\right] .
\end{equation*}
Let $\mathbb{\tilde{E}}_{t}\left[ X\right] =\mathbb{E}\left[ j_{t}\left(
X\right) \right] $, then by lemma \ref{unitary rotations} 
\begin{equation*}
\pi _{t}\left( X\right) =\mathbb{E}\left[ j_{t}\left( X\right) \mid \mathfrak{Y}%
_{t]}^{\mathrm{out}}\right] =V\left( t\right) ^{\dag }\mathbb{\tilde{E}}_{t}%
\left[ X|\mathfrak{Y}_{t]}^{\mathrm{in}}\right] V\left( t\right) .
\end{equation*}

\subsubsection{Reference Probability Approach}

Suppose that there is an adapted process $F\left( \cdot \right) $ such that $%
F\left( t\right) \in \left( \mathfrak{Y}_{t]}^{\mathrm{in}}\right) ^{\prime }$
and $\mathbb{\tilde{E}}_{t}\left[ X|\right] =\mathbb{E}\left[ F\left(
t\right) ^{\dag }\left( X\otimes 1\right) F\left( t\right) \right] $ for all
system operators $X$, then by lemma \ref{QBayes}

\begin{eqnarray*}
\pi _{t}\left( X\right) &=&\mathbb{E}\left[ j_{t}\left( X\right) \mid \mathfrak{Y%
}_{t]}^{\mathrm{out}}\right] \\
&=&V\left( t\right) ^{\dag }\mathbb{\tilde{E}}_{t}\left[ X\mid \mathfrak{Y}%
_{t]}^{\mathrm{in}}\right] V\left( t\right) \\
&=&V\left( t\right) ^{\dag }\frac{\mathbb{E}\left[ F\left( t\right) ^{\dag
}\left( X\otimes 1\right) F\left( t\right) \mid \mathfrak{Y}_{t]}^{\mathrm{in}}%
\right] }{\mathbb{E}\left[ F\left( t\right) ^{\dag }F\left( t\right) \mid 
\mathfrak{Y}_{t]}^{\mathrm{in}}\right] }V\left( t\right)
\end{eqnarray*}
This is essentially a non-commutative version of the Girsanov transformation
from stochastic analysis. The essential feature is that the transformation
operators $F\left( t\right) $ giving the change of representation for the
expectation lie in the commutant of the measurement algebra up to time $t$.

\bigskip

We therefore obtain an operator-valued Kallianpur-Striebel relation 
\begin{equation}
\pi _{t}\left( X\right) =\frac{\sigma _{t}\left( X\right) }{\sigma
_{t}\left( 1\right) },  \label{QKaSt}
\end{equation}
which may be called the \emph{quantum Kallianpur-Striebel}, where 
\begin{equation}
\sigma _{t}\left( X\right) := V\left( t\right) ^{\dag }\mathbb{E}\left[
F\left( t\right) ^{\dag }\left( X\otimes 1\right) F\left( t\right) \mid 
\mathfrak{Y}_{t]}^{\mathrm{in}}\right] V\left( t\right) .  \label{sigma}
\end{equation}

\section{Coherent State Filters}

We shall consider the class of states 
\begin{equation}
\mathbb{E}^{\beta }\left[ \,\cdot \,\right] =\langle \psi ^{\beta }|\cdot
\psi ^{\beta }\rangle 
\end{equation}
of the form 
\begin{equation}
\psi ^{\beta }=\phi \otimes \Psi \left( \beta \right) 
\end{equation}
where $\phi $ is a normalized vector in the system Hilbert space and $\Psi
\left( \beta \right) $ is the coherent state with test function $\beta \in
L^{2}[0,\infty )$. We note that 
\begin{eqnarray*}
dB\left( t\right) \;\Psi \left( \beta \right)  &=&\beta \left( t\right)
dt\;\Psi \left( \beta \right) , \\
d\Lambda \left( t\right) \;\Psi \left( \beta \right)  &=&\beta \left(
t\right) dB^{\ast }(t)\;\Psi \left( \beta \right) .
\end{eqnarray*}

We see that 
\begin{equation}
\mathbb{E}^{\beta }[dj_{t}\left( X\right) ]=\mathbb{E}^{\beta }[j_{t}\left( 
\mathcal{L}^{\beta \left( t\right) }X\right) ]dt  \label{averaged}
\end{equation}
where 
\begin{equation}
\mathcal{L}^{\beta \left( t\right) }X=\mathcal{L}_{00}X+\beta \left(
t\right) ^{\ast }\mathcal{L}_{10}X+\beta \left( t\right) \mathcal{L}%
_{01}X+|\beta \left( t\right) |^{2}\mathcal{L}_{11}X
\end{equation}
The generator is again of Lindblad form and in particular we have 
$$\mathcal{L%
}^{\beta \left( t\right) }\equiv \mathcal{L}_{\left( L^{\beta \left( t\right) },H^{\beta \left( t\right) } \right) }$$ 
with 
\begin{equation}
L^{\beta \left( t\right) }=S\beta \left( t\right) +L,
\quad
H^{\beta \left( t\right) }= \frac{1}{2i} \left( L^\dag \beta (t) -L \beta (t)^\ast \right)   .  \label{L_t}
\end{equation}
We may define a parametrized family of density matrices on the system by
setting tr$_{\mathfrak{h}}\left\{ \varrho _{t}X\right\} =\mathbb{E}^{\beta
}[j_{t}\left( X\right) ]$, in which case we deduce the master equation 
\begin{equation*}
\dot{\varrho}_{t}=\mathcal{L}^{\beta \left( t\right) \prime }\left( \varrho
\right) ,
\end{equation*}
where the adjoint is defined through the duality tr$_{\mathfrak{h}%
}\left\{ \varrho \,\mathcal{L}X\right\} =$tr$_{\mathfrak{h}}\left\{ \mathcal{L}%
^{\prime }\varrho \,X\right\} $.

\bigskip

From the input-output relation for the field 
\begin{equation*}
dB_{\mathrm{out}}=j_{t}\left( S\right) dB\left( t\right) +j_{t}\left(
L\right) dt
\end{equation*}
we obtain the average 
\begin{equation*}
\mathbb{E}^{\beta }\left[ dB_{\mathrm{out}}\right]  =\left\{ j_{t}\left(
S\right) \beta \left( t\right) +j_{t}\left( L\right) \right\} dt
=j_{t}\left( L^{\beta \left( t\right) }\right) dt.
\end{equation*}

\subsection{Quadrature Measurement}

We take $Y^{\mathrm{in}}\left( t\right) =B\left( t\right) +B^{\ast }\left(
t\right) $ which is a quadrature of the input field. Setting $\psi \left(
t\right) =V\left( t\right) \psi $ we have that 
\begin{equation}
d\psi \left( t\right) =\left[\left( S-1\right) \beta \left( t\right)
+L\right] dB^{\ast }\left( t\right) \psi \left( t\right) - \left[ L^{\ast
}S\beta \left( t\right) +\frac{1}{2}L^{\ast }L+iH \right] dt\,\psi \left(
t\right) 
\label{apply Holevo trick}
\end{equation}
At this stage we apply a trick which is essentially a \emph{quantum Girsanov transformation}. This trick is due to Belavkin \cite{Bel89} and Holevo \cite{Hol90}.
We now add a term proportional to $dB\left( t\right) \psi \left( t\right) $
to get

\begin{eqnarray*}
d\psi \left( t\right)  &=&\left( \left( S-1\right) \beta \left( t\right)
+L\right) \left[ dB^{\ast }\left( t\right) +dB\left( t\right) \right] \psi
\left( t\right)  \\
&&-\left( L^{\ast }S\beta \left( t\right) +\frac{1}{2}L^{\ast }L+iH+\left(
\left( S-1\right) \beta \left( t\right) +L\right) \beta \left( t\right)
\right) dt\,\psi \left( t\right)  \\
&\equiv &\tilde{L}_{t}dY^{\mathrm{in}}\left( t\right) \psi \left( t\right) +%
\tilde{K}_{t}dt\,\psi \left( t\right) ,
\end{eqnarray*}
where 
\begin{eqnarray*}
\tilde{L}_{t} &=&L+\left( S-I\right) \beta \left( t\right) =L^{\beta \left(
t\right) }-\beta \left( t\right) , \\
\tilde{K}_{t} &=&-L^{\ast }S\beta \left( t\right) -\frac{1}{2}L^{\ast
}L-iH-L^{\beta \left( t\right) }\beta \left( t\right) +\beta \left( t\right)
^{2}.
\end{eqnarray*}
It follows that $\psi \left( t\right) \equiv F\left( t\right) \psi $ where $%
F\left( t\right) $ is the adapted process satisfying the QSDE 
\begin{equation*}
dF\left( t\right) =\tilde{L}_{t}dY^{\mathrm{in}}\left( t\right) F\left(
t\right) +\tilde{K}_{t}dt\,F\left( t\right) ,\;F\left( 0\right) =I.
\end{equation*}
Moreover $F\left( t\right) $ is in the commutant of $\mathfrak{Y}_{t]}^{\mathrm{%
in}}$ and therefore allows us to perform the non-commutative Girsanov trick.

\bigskip

From the quantum It\={o} product rule we then see that 
\begin{equation*}
d\left[ F^{\ast }\left( t\right) XF\left( t\right) \right]  = F^{\ast
}\left( t\right) \left( X\tilde{L}_{t}+\tilde{L}_{t}^{\ast }X\right) F\left(
t\right) dY^{\mathrm{in}}\left( t\right)  \\
+F^{\ast }\left( t\right) \left( \tilde{L}_{t}^{\ast }X\tilde{L}_{t}+X%
\tilde{K}_{t}+\tilde{K}_{t}X\right) F\left( t\right) dt
\end{equation*}
and this leads to the SDE for the un-normalized filter 
\begin{equation*}
d\sigma _{t}\left( X\right)  = \sigma _{t}\left( X\tilde{L}_{t}+\tilde{L}
_{t}^{\ast }X\right) dY^{\mathrm{out}}\left( t\right)  \\
+\sigma _{t}\left( \tilde{L}_{t}^{\ast }X\tilde{L}_{t}+X\tilde{K}_{t}+
\tilde{K}_{t}X\right) dt.
\end{equation*}
After a small bit of algebra, this may be written in the form 
\begin{equation*}
d\sigma _{t}\left( X\right) =\sigma _{t}\left( X\tilde{L}_{t}+\tilde{L}%
_{t}^{\ast }X\right) \left[ dY^{\mathrm{out}}\left( t\right) -\left( \beta
\left( t\right) +\beta \left( t\right) ^{\ast }\right) dt\right] +\sigma
_{t}\left( \mathcal{L}^{\beta \left( t\right) }X\right) dt.
\end{equation*}
This is the quantum Zakai equation for the filter based on continuous
measurement of the output field quadrature.

\bigskip

To obtain the quantum Kushner-Stratonovich equation we first observe that
the normalization satisfies the SDE 
\begin{equation*}
d\sigma _{t}\left( I\right) =\sigma _{t}\left( \tilde{L}_{t}+\tilde{L}%
_{t}^{\ast }\right) \left[ dY^{\mathrm{out}}\left( t\right) -\left( \beta
\left( t\right) +\beta \left( t\right) ^{\ast }\right) dt\right] 
\end{equation*}
and by It\={o}'s formula 
\begin{equation*}
d\frac{1}{\sigma _{t}\left( I\right) }=-\frac{\sigma _{t}\left( \tilde{L}%
_{t}+\tilde{L}_{t}^{\ast }\right) }{\sigma _{t}\left( I\right) ^{2}}\left[
dY^{\mathrm{out}}\left( t\right) -\left( \beta \left( t\right) +\beta \left(
t\right) ^{\ast }\right) dt\right] +\frac{\sigma _{t}\left( \tilde{L}_{t}+%
\tilde{L}_{t}^{\ast }\right) ^{2}}{\sigma _{t}\left( I\right) ^{3}}dt.
\end{equation*}
The product rule then allows us to determine the SDE for $\pi _{t}\left(
X\right) =\dfrac{\sigma _{t}\left( X\right) }{\sigma _{t}\left( I\right) }$: 
\begin{equation*}
d\pi _{t}\left( X\right) =\pi _{t}\left( \mathcal{L}^{\beta \left( t\right)
}X\right) dt +\left\{ \pi _{t}\left( X\tilde{L}_{t}+\tilde{L}_{t}^{\ast
}X\right) -\pi _{t}\left( X\right) \pi _{t}\left( \tilde{L}_{t}+\tilde{L}%
_{t}^{\ast }\right) \right\} dI\left( t\right) ,
\end{equation*}
where the innovations process satisfies 
\begin{equation*}
dI\left( t\right) =dY^{\mathrm{out}}\left( t\right) -\left[ \pi _{t}\left( 
\tilde{L}_{t}+\tilde{L}_{t}^{\ast }\right) +\beta \left( t\right) +\beta
\left( t\right) ^{\ast }\right] dt.
\end{equation*}
We note that the innovations is martingale with respect to filtration
generated by the output process for the choice of probability measure
determined by the coherent state.

\subsubsection{The Quadrature Measurement Filter for a Coherent state}

In it convenient to write the filtering equations in terms of the operators $%
L^{\beta \left( t\right) }$. The result is the \emph{Belavkin-Kushner-Stratonovich equation} for the filtered estimate based on  optimal estimation of continuous non-demolition field-quadrature measurements in a coherent state $\beta (\cdot )$.
\begin{equation}
d\pi _{t}\left( X\right) = \pi _{t}\left( \mathcal{L}^{\beta \left( t\right)
}X\right) dt
 +\left\{ \pi _{t}\left( XL^{\beta \left( t\right) }+L^{\beta
\left( t\right) \ast }X\right)-\pi _{t}\left( X\right) \pi _{t}\left(
L^{\beta \left( t\right) }+L^{\beta \left( t\right) \ast }\right) \right\}
dI^{\mathrm{quad}}\left( t\right) ,  \notag
\\
\label{quadfilter}
\end{equation}
with the innovations 
\begin{equation}
dI^{\mathrm{quad}}\left( t\right) =dY^{\mathrm{out}}\left( t\right) -\pi
_{t}\left( L^{\beta \left( t\right) }+L^{\beta \left( t\right) \ast }\right)
dt.
\end{equation}

\subsection{Photon Counting Measurement}

For convenience we shall derive the filter based on measuring the number of
output photons under the assumption that the function $\beta $ is bounded
away from zero. We discuss how this restriction can be removed later. We now
set $Y^{\mathrm{in}}\left( t\right) =\Lambda \left( t\right) $. We again
seek to construct an adapted process $F\left( t\right) $ in the commutant of 
$\mathfrak{Y}_{t]}^{\mathrm{in}}$ such that $\psi \left( t\right) =F\left(
t\right) \psi $. We start with \ref{apply Holevo trick} again but now note that 
\begin{equation*}
dB^{\ast }\left( t\right) \psi \left( t\right) =\frac{1}{\beta \left(
t\right) }d\Lambda (t)\psi \left( t\right)
\end{equation*}
and making this substitution gives 
\begin{equation*}
d\psi \left( t\right) =\frac{1}{\beta \left( t\right) }\tilde{L}_{t}dY^{%
\mathrm{in}}(t)\psi \left( t\right) -\left( \frac{1}{2}L^{\ast }L+iH+L^{\ast
}S\beta \left( t\right) \right) dt\,\psi \left( t\right) .
\end{equation*}
We are then lead to the Zakai equation 
\begin{multline*}
d\sigma _{t}\left( X\right) =\frac{1}{|\beta \left( t\right) |^{2}}\sigma
_{t}\left( \tilde{L}_{t}^{\ast }X\tilde{L}_{t}+\beta (t)^{\ast }X\tilde{L}%
_{t}+\tilde{L}_{t}^{\ast }X\beta \left( t\right) \right) dY^{\mathrm{out}%
}\left( t\right) \\
-\sigma _{t}\left( \frac{1}{2}XL^{\ast }L+\frac{1}{2}L^{\ast }LX+i\left[ X,H%
\right] -XL^{\ast }S\beta \left( t\right) -\beta \left( t\right) ^{\ast
}S^{\ast }LX\right) dt
\end{multline*}
which may be rearranged as 
\begin{equation*}
d\sigma _{t}\left( X\right) =\sigma _{t}\left( \mathcal{L}^{\beta }X\right)
dt 
+\frac{1}{|\beta \left( t\right) |^{2}}\sigma _{t}\left( \tilde{L}_{t}^{\ast
}X\tilde{L}_{t}+\beta (t)^{\ast }X\tilde{L}_{t}+\tilde{L}_{t}^{\ast }X\beta
\left( t\right) \right) \left[ dY^{\mathrm{out}}\left( t\right) -|\beta
\left( t\right) |^{2}dt\right] .
\end{equation*}

To determine the normalized filter, we note that 
\begin{equation*}
d\sigma _{t}\left( I\right) =\frac{1}{|\beta \left( t\right) |^{2}}\sigma
_{t}\left( \tilde{L}_{t}^{\ast }\tilde{L}_{t}+\beta (t)^{\ast }\tilde{L}_{t}+%
\tilde{L}_{t}^{\ast }\beta \left( t\right) \right) \left[ dY^{\mathrm{out}%
}\left( t\right) -|\beta \left( t\right) |^{2}dt\right]
\end{equation*}
and that 
\begin{eqnarray*}
d\frac{1}{\sigma _{t}\left( I\right) } &=&\frac{\sigma _{t}\left( \tilde{L}%
_{t}^{\ast }\tilde{L}_{t}+\beta (t)^{\ast }\tilde{L}_{t}+\tilde{L}_{t}^{\ast
}\beta \left( t\right) \right) }{\sigma _{t}\left( I\right) ^{2}}dt \\
&&-\frac{\sigma _{t}\left( \tilde{L}_{t}^{\ast }\tilde{L}_{t}+\beta
(t)^{\ast }\tilde{L}_{t}+\tilde{L}_{t}^{\ast }\beta \left( t\right) \right) 
}{\sigma _{t}\left( I\right) \left[ \sigma _{t}\left( \tilde{L}_{t}^{\ast }%
\tilde{L}_{t}+\beta (t)^{\ast }\tilde{L}_{t}+\tilde{L}_{t}^{\ast }\beta
\left( t\right) \right) +|\beta \left( t\right) |^{2}\right] }dY^{\mathrm{out%
}}\left( t\right) .
\end{eqnarray*}
An application of the It\={o} product formula then yields the quantum analogue of the Kushner-Stratonovich
equation for the normalized filter; 
\begin{multline*}
d\pi _{t}\left( X\right) =\pi _{t}\left( \mathcal{L}^{\beta }X\right) dt+%
\frac{1}{\pi _{t}\left( \tilde{L}_{t}^{\ast }\tilde{L}_{t}+\beta (t)^{\ast }%
\tilde{L}_{t}+\tilde{L}_{t}^{\ast }\beta \left( t\right) \right) +|\beta
\left( t\right) |^{2}}dI\left( t\right) \\
\times \left\{ \pi _{t}(\tilde{L}_{t}^{\ast }X\tilde{L}_{t}+\beta (t)^{\ast
}X\tilde{L}_{t}+\tilde{L}_{t}^{\ast }X\beta \left( t\right) )-\pi _{t}\left(
X\right) \pi _{t}(\tilde{L}_{t}^{\ast }\tilde{L}_{t}+\beta (t)^{\ast }\tilde{%
L}_{t}+\tilde{L}_{t}^{\ast }\beta \left( t\right) )\right\}
\end{multline*}
and the innovations process is now 
\begin{equation*}
dI\left( t\right) =dY^{\mathrm{out}}\left( t\right) -\left[ \pi _{t}\left( 
\tilde{L}_{t}^{\ast }\tilde{L}_{t}+\beta (t)^{\ast }\tilde{L}_{t}+\tilde{L}%
_{t}^{\ast }\beta \left( t\right) \right) +|\beta \left( t\right) |^{2}%
\right] dt.
\end{equation*}
We note that the innovations is again a martingale with respect to
filtration generated by the output process for the choice of probability
measure determined by the coherent state.

\indent
The derivation above relied on the assumption that $\beta \left( t\right)
\neq 0$, however this is not actually essential. In the case of a vacuum
input, it is possible to apply an additional rotation $W\left( t\right) $
satisfying $dW\left( t\right) =[z^{\ast }dB\left( t\right) -zdB^{\ast
}\left( t\right) -\frac{1}{2}|z|^{2}dt]W\left( t\right)$ , with $W\left( 0\right) =I$,
 and apply the reference probability technique to the von Neumann algebra
generated by $N\left( t\right) =W\left( t\right) \Lambda \left( t\right)
W\left( t\right) ^{\ast }$: this leads to a Zakai equation that explicitly
depends on the choice of $z\in C$ however the Kushner-Stratonovich equation
for the normalized filter will be $z$-independent. Similarly for the general
coherent state considered here, we could take $z$ to be a function of $t$ in
which case we must chose $z$ so that $\beta \left( t\right) +z\left(
t\right) \neq 0$. The Kushner-Stratonovich equation obtained will then be
identical to what we have just derived.

\subsubsection{Photon Counting Measurement in a Coherent State}

Again it is convenient to re-express the filter in term of $L^{\beta \left(
t\right) }$. We now obtain the Belavkin-Kushner-Stratonovich equation for the filtered estimate based on  optimal estimation of continuous non-demolition field-quanta number measurements in a coherent state $\beta (\cdot )$.
\begin{equation}
d\pi _{t}\left( X\right) =\pi _{t}\left( \mathcal{L}^{\beta \left( t\right)
}X\right) dt+\left\{ \frac{\pi _{t}\left( L^{\beta \left( t\right) \ast
}XL^{\beta \left( t\right) }\right) }{\pi _{t}\left( L^{\beta \left(
t\right) \ast }L^{\beta \left( t\right) }\right) }-\pi _{t}\left( X\right)
\right\} dI^{\mathrm{num}}\left( t\right) ,  \label{numberfilter}
\end{equation}
with the innovations 
\begin{equation}
dI^{\mathrm{num}}\left( t\right) =dY^{\mathrm{out}}\left( t\right) -\pi
_{t}\left( L^{\beta \left( t\right) \ast }L^{\beta \left( t\right) }\right)
dt.
\end{equation}

\section{Characteristic Function Approach}

As an alternative to the reference probability approach, we apply a method
based on introducing a process $C\left( t\right) $ satisfying the QSDE
\begin{equation}
dC\left( t\right) =f\left( t\right) C\left( t\right) dY\left( t\right) ,
\end{equation}
with initial condition $C\left( 0\right) =I$. Here we assume that $f$ is
integrable, but otherwise arbitrary. The technique is to make an ansatz of
the form
\begin{equation}
d\pi _{t}\left( X\right) =\mathcal{F}_{t}\left( X\right) dt+\mathcal{H}%
_{t}\left( X\right) dY\left( t\right)   \label{filter ansatz}
\end{equation}
where we assume that the processes $\mathcal{F}_{t}\left( X\right) $ and $%
\mathcal{H}_{t}\left( X\right) $ are adapted and lie in $\mathfrak{Y}_{t]}$.
These coefficients may be deduced from the identity
\begin{equation*}
\mathbb{E}\left[ \left( \pi _{t}\left( X\right) -j_{t}\left( X\right)
\right) C\left( t\right) \right] =0
\end{equation*}
which is valid since $C\left( t\right) \in \mathfrak{Y}_{t]}$. We note that the
It\={o} product rule implies $I+II+III=0$ where

\begin{eqnarray*}
I &=&\mathbb{E}\left[ \left( d\pi _{t}\left( X\right) -dj_{t}\left( X\right)
\right) C\left( t\right) \right] , \\
II &=&\mathbb{E}\left[ \left( \pi _{t}\left( X\right) -j_{t}\left( X\right)
\right) dC\left( t\right) \right] , \\
III &=&\mathbb{E}\left[ \left( d\pi _{t}\left( X\right) -dj_{t}\left(
X\right) \right) dC\left( t\right) \right] .
\end{eqnarray*}

We illustrate how this works in the case of quadrature and photon counting
in a coherent state. For convenience of notation we shall write $S_{t}$ for $%
j_{t}\left( S\right) $, etc.

\subsection{Quadrature Measurement}

Here we have
\begin{equation*}
dY\left( t\right) =S_{t}dB\left( t\right) +S_{t}^{\ast }dB\left( t\right)
^{\ast }+\left( L_{t}+L_{t}^{\ast }\right) dt
\end{equation*}
so that
\begin{eqnarray*}
I &=&\mathbb{E}_{\beta }\left[ \mathcal{F}_{t}\left( X\right) C\left(
t\right) +\mathcal{H}_{t}\left( X\right) \left( S_{t}\beta _{t}+S_{t}^{\ast
}\beta _{t}^{\ast }+L_{t}+L_{t}^{\ast }\right) C\left( t\right) \right] dt \\
&&-\mathbb{E}_{\beta }\left[ \left\{ \left( \mathcal{L}_{00}X\right)
_{t}+\left( \mathcal{L}_{01}X\right) _{t}\beta _{t}+\left( \mathcal{L}%
_{10}X\right) _{t}\beta _{t}^{\ast }+\left( \mathcal{L}_{11}X\right)
_{t}\left| \beta _{t}\right| ^{2}\right\} C\left( t\right) \right] dt, \\
II &=&\mathbb{E}_{\beta }\left[ \left( \pi _{t}\left( X\right) -X_{t}\right)
f\left( t\right) C\left( t\right) \left( S_{t}\beta _{t}+S_{t}^{\ast }\beta
_{t}^{\ast }+L_{t}+L_{t}^{\ast }\right) \right] dt, \\
III &=&\mathbb{E}_{\beta }\left[ \left\{ \mathcal{H}_{t}\left( X\right)
-\left( \mathcal{L}_{01}X\right) _{t}S_{t}^{\ast }\beta _{t}^{\ast }-\left( 
\mathcal{L}_{11}X\right) _{t}S_{t}^{\ast }\beta _{t}^{\ast }\right\} f\left(
t\right) C\left( t\right) \right] dt.
\end{eqnarray*}
Now from the identity $I+II+III=0$ we may extract separately the
coefficients of $f\left( t\right) C\left( t\right) $ and $C\left( t\right) $
as $f\left( t\right) $ was arbitrary to deduce
\begin{eqnarray*}
\pi _{t}\left( \left( \pi _{t}\left( X\right) -X_{t}\right) \left(
S_{t}\beta _{t}+S_{t}^{\ast }\beta _{t}^{\ast }+L_{t}+L_{t}^{\ast }\right) \right)  
+ \pi_t \left( \mathcal{H}_{t}\left( X\right) -\left( \mathcal{L}_{01}X\right)
_{t}S_{t}^{\ast }\beta _{t}^{\ast }-\left( \mathcal{L}_{11}X\right)
_{t}S_{t}^{\ast }\beta _{t}^{\ast }\right)  &=&0, \\
\pi _{t}\left( \mathcal{F}_{t}\left( X\right) +\mathcal{H}_{t}\left(
X\right) \left( S_{t}\beta _{t}+S_{t}^{\ast }\beta _{t}^{\ast
}+L_{t}+L_{t}^{\ast }\right) -\left( \mathcal{L}^{\beta \left( t\right)
}X\right) _{t}\right)  &=&0.
\end{eqnarray*}
Using the projective property of the conditional expectation $\left( \pi
_{t}\circ \pi _{t}=\pi _{t}\right) $ and the assumption that $\mathcal{F}%
_{t}\left( X\right) $ and $\mathcal{H}_{t}\left( X\right) $ lie in $\mathfrak{Y}%
_{t]}$, we find after a little algebra that
\begin{eqnarray*}
\mathcal{H}_{t}\left( X\right)  &=&\pi _{t}\left( XL^{\beta \left( t\right)
}+L^{\beta \left( t\right) \ast }X\right) -\pi _{t}\left( X\right) \pi
_{t}\left( L^{\beta \left( t\right) }+L^{\beta \left( t\right) \ast }\right)
, \\
\mathcal{F}_{t}\left( X\right)  &=&\pi _{t}\left( \mathcal{L}^{\beta \left(
t\right) }X\right) -\mathcal{H}_{t}\left( X\right) \pi _{t}\left( L^{\beta
\left( t\right) }+L^{\beta \left( t\right) \ast }\right) ,
\end{eqnarray*}
so that the equation (\ref{filter ansatz}) reads as
\begin{equation*}
d\pi _{t}\left( X\right) =\pi _{t}\left( \mathcal{L}^{\beta \left( t\right)
}X\right) dt+\mathcal{H}_{t}\left( X\right) \left[ dY\left( t\right) -\pi
_{t}\left( L^{\beta \left( t\right) }+L^{\beta \left( t\right) \ast }\right)
dt\right] .
\end{equation*}

\subsection{Photon Counting Measurement}

We now have
\begin{equation*}
dY\left( t\right) =d\Lambda \left( t\right) +L_{t}^{\ast }S_{t}dB\left(
t\right) +S_{t}^{\ast }L_{t}dB\left( t\right) ^{\ast }+L_{t}^{\ast }L_{t}dt
\end{equation*}
so that
\begin{eqnarray*}
I &=&\mathbb{E}_{\beta }\left[ \left\{ \mathcal{F}_{t}\left( X\right)  +\mathcal{H}_{t}\left( X\right) \left( \left| \beta _{t}\right|
^{2}+L_{t}^{\ast }S_{t}\beta _{t}+S_{t}^{\ast }L_{t}\beta _{t}^{\ast
}+L_{t}^{\ast }L_{t}\right)  \right\} C\left( t\right) \right] dt 
 \mathbb{E}_{\beta }\left[ \left\{  -\left( \mathcal{L}^{\beta
\left( t\right) }X\right) _{t} \right\} C\left( t\right) \right] dt \\
II &=&\mathbb{E}_{\beta }\left[ \left( \pi _{t}\left( X\right) -X_{t}\right)
f\left( t\right) C\left( t\right) \left( \left| \beta _{t}\right|
^{2}+L_{t}^{\ast }S_{t}\beta _{t}+S_{t}^{\ast }L_{t}\beta _{t}^{\ast
}+L_{t}^{\ast }L_{t}\right) \right] dt, \\
III &=&\mathbb{E}_{\beta }\left[ \mathcal{H}_{t}\left( X\right) \left(
\left| \beta _{t}\right| ^{2}+L_{t}^{\ast }S_{t}\beta _{t}+S_{t}^{\ast
}L_{t}\beta _{t}^{\ast }+L_{t}^{\ast }L_{t}\right) f\left( t\right) C\left(
t\right) \right] dt \\
&&-\mathbb{E}_{\beta }\left[ \left\{ \left( \mathcal{L}_{01}X\right)
_{t}S_{t}^{\ast }L_{t}+\left( \mathcal{L}_{01}X\right) _{t}\beta _{t} \right\} f\left(
t\right) C\left( t\right) \right] dt\\
&&-\mathbb{E}_{\beta }\left[ \left\{ \left( 
\mathcal{L}_{11}X\right) _{t}\left| \beta _{t}\right| ^{2}+\left( \mathcal{L}%
_{11}X\right) _{t}S_{t}^{\ast }L_{t}\beta _{t}^{\ast } \right\} f\left(
t\right) C\left( t\right) \right] dt.
\end{eqnarray*}
This time, the identity $I+II+III=0$ implies
\begin{eqnarray*}
\pi _{t}\left( \left( \pi _{t}\left( X\right) -X_{t}\right) L_{t}^{\beta
\left( t\right) \ast }L_{t}^{\beta \left( t\right) }+\mathcal{H}_{t}\left(
X\right) L_{t}^{\beta \left( t\right) \ast }L_{t}^{\beta \left( t\right)
}\right) 
-\pi_t \left( \left( \mathcal{L}_{01}X\right) _{t}S_{t}^{\ast }L_{t}^{\beta \left(
t\right) }+\left( \mathcal{L}_{11}X\right) _{t}S_{t}^{\ast }L_{t}^{\beta
\left( t\right) }\beta _{t}^{\ast }\right)  &=&0, \\
\pi _{t}\left( \mathcal{F}_{t}\left( X\right) +\mathcal{H}_{t}\left(
X\right) L_{t}^{\beta \left( t\right) \ast }L_{t}^{\beta \left( t\right)
}-\left( \mathcal{L}^{\beta \left( t\right) }X\right) _{t}\right)  &=&0.
\end{eqnarray*}
Again, after a little algebra, we find that
\begin{eqnarray*}
\mathcal{H}_{t}\left( X\right)  &=&\frac{\pi _{t}\left( L^{\beta \left(
t\right) \ast }XL^{\beta \left( t\right) }\right) }{\pi _{t}\left( L^{\beta
\left( t\right) \ast }L^{\beta \left( t\right) }\right) }-\pi _{t}\left(
X\right) , \\
\mathcal{F}_{t}\left( X\right)  &=&\pi _{t}\left( \mathcal{L}^{\beta \left(
t\right) }X\right) -\mathcal{H}_{t}\left( X\right) \pi _{t}\left( L^{\beta
\left( t\right) \ast }L^{\beta \left( t\right) }\right) ,
\end{eqnarray*}
so that the equation (\ref{filter ansatz}) reads as
\begin{equation*}
d\pi _{t}\left( X\right) =\pi _{t}\left( \mathcal{L}^{\beta \left( t\right)
}X\right) dt+\mathcal{H}_{t}\left( X\right) \left[ dY\left( t\right) -\pi
_{t}\left( L^{\beta \left( t\right) \ast }L^{\beta \left( t\right) }\right)
dt\right] .
\end{equation*}

In both cases, the form of the filter is identical to what we found using
the reference probability approach.

\section{Conclusion}

Both the quadrature filter (\ref{quadfilter}) and the photon counting filter
(\ref{numberfilter}) take on the same form as in to the vacuum case and of course reduce to these filters when we set $%
\beta \equiv 0$. In both cases it is clear that averaging over the output
gives 
\begin{equation*}
\mathbb{E}^{\beta }\left[ d\pi _{t}\left( X\right) \right] =\mathbb{E}%
^{\beta }\left[ j_{t}\left( \mathcal{L}^{\beta (t)}X\right) \right] dt
\end{equation*}
which is clearly the correct unconditioned dynamics in agreement with (\ref
{averaged}), and we obtain the correct master equation.

\bigskip

It is worth commenting on the fact that the pair of equations now replacing
the dynamical and observation relations (\ref{pairsys},\ref{pairobs}) are
the Heisenberg-Langevin equation (\ref{dynamical}) and the appropriate
component of the input-output relation (\ref{output}). The process and
observation noise have the same origin however the nature of the quantum
filtering based on a non-demolition measurement scheme results in a set of
equations that resemble the uncorrelated classical Kushner-Stratonovich
equations.

\subsection{Is Quantum Filtering still a Pure Filtering Problem?}
The form of the input-output relations (\ref{output}) might suggest that it
is possible to learn something about the system dynamics by examining the
quadratic variation of the output process, however, this is not the case! We
in fact have an enforced ``too good to be true'' situation here as the output
fields satisfy the same canonical commutation relations as the inputs with
the result that the quantum It\={o} table for the output processes $B_{%
\mathrm{out}}$, $B_{\mathrm{out}}^{\ast }$ and $\Lambda _{\mathrm{out}}$ has
precisely the same structure as table \ref{QIT}. Therefore we always deal
with a pure filtering theory in the quantum models considered here.

\section*{Acknowledgments}

\par\bigskip\noindent
{\bf Acknowledgment.} 
The authors gratefully acknowledge the support of the UK Engineering and Physical Sciences Research Council under grant EP/G039275/1.

\bibliographystyle{amsplain}

\end{document}